\documentclass[12pt]{article}       
\usepackage{graphicx}
\usepackage{calc}
\usepackage{array}
\usepackage{graphicx}
\usepackage{subfigure}
\usepackage[indent,bf]{caption}
\usepackage{rotating}
\usepackage{setspace}
\usepackage{longtable}
\usepackage{mdframed}
\usepackage{rotating}
\usepackage{hyperref}

\usepackage{amssymb, amsmath, amsfonts, amsthm, epsfig}
\usepackage{epsfig,latexsym}
\usepackage{multicol, paralist}
\usepackage{pgf,tikz}
\usetikzlibrary{arrows,snakes,backgrounds, automata}
\usepackage{verbatim}
\usepackage{url}
\usepackage{fancyhdr}
\usepackage{authblk}
\usepackage{marginnote}

\usepackage{geometry}

\usetikzlibrary{arrows}

\newtheorem{thm}{Theorem}

\newtheorem{conj}[thm]{Conjecture}

\newtheorem{lem}[thm]{Lemma}

\newcommand{\edom}{\gamma^\infty}

\begin{document}


\title{Eternal domination on prisms of graphs}



\author[1]{Aaron Krim-Yee}
\affil[1]{Department of Bioengineering, McGill University, Montreal, QC, Canada}
\author[2,3]{Ben Seamone}
\affil[2]{Mathematics Department, Dawson College, Montreal, QC, Canada}
\affil[3]{D\'{e}partement d'informatique et de recherche op\'{e}rationnelle, Universit\'{e} de Montr\'{e}al, Montreal, QC, Canada}
\author[4]{Virg\'elot Virgile}
\affil[4]{D\'epartement de math\'ematiques et de statistique, Universit\'{e} de Montr\'{e}al, Montreal, QC, Canada}
\maketitle

\begin{abstract}
An eternal dominating set of a graph $G$ is a set of vertices (or ``guards'') which dominates $G$ and which can defend any infinite series of vertex attacks, where an attack is defended by moving one guard along an edge from its current position to the attacked vertex.  The size of the smallest eternal dominating set is denoted $\gamma^\infty(G)$ and is called the eternal domination number of $G$.  In this paper, we answer a conjecture of Klostermeyer and Mynhardt [Discussiones Mathematicae Graph Theory, vol. 35, pp. 283-300], showing that there exist there are infinitely many graphs $G$ such that $\gamma^\infty(G)=\theta(G)$ and $\gamma^\infty(G \Box K_2)<\theta(G \Box K_2)$, where $\theta(G)$ denotes the clique cover number of $G$.
\end{abstract}

\section{Introduction}

Let $G = (V,E)$ be a graph, which is assuming throughout to be simple and finite.  A set $X \subseteq V(G)$ is called a \textit{dominating set} if for every $u \in V(G) \setminus X$ there exists an $x \in X$ such that $ux \in E(G)$.  Consider the following graph security model.  A set of guards begins by occupying a dominating set $X$ in a graph $G$, and must respond to an infinite sequence of attacks.  By this, we mean that after a vertex $u \notin X$ is chosen by an attacker, one guard which is adjacent to $u$ must move from its current position to $u$; necessarily, the resulting set of positions must still be a dominating set of $G$.  If $k$ guards can move in this way to respond to any infinite sequence of attacks, then we say $G$ can be \textit{eternally $k$-guarded}.  The minimum $k$ for which $G$ can be eternally $k$-guarded is called the \textit{eternal domination number} of $G$, which is denoted $\edom(G)$.

The study of $\edom(G)$ was introduced in \cite{burger2004infinite}, where (among other topics) its relation to other graph parameters was studied.  Recall that an \textit{independent set} of $G$ is a set of pairwise non-adjacent vertices, and $\alpha(G)$, called the \textit{independence number} of $G$, is the maximum cardinality of an independent set in $G$.  A \textit{clique} of $G$ is a set of pairwise adjacent vertices.  A \textit{clique cover} of $G$ is a set $\{X_1,\ldots,X_k\}$ for which each $X_i$ is a clique and $\cup_{i=1}^kX_i = V(G)$.  The cardinality of a minimum clique cover of $G$ is called the \textit{clique cover number} of $G$ and is denoted $\theta(G)$.  It was shown in \cite{burger2004infinite} that, for any graph $G$, $\alpha(G) \leq \gamma^\infty(G) \leq \theta(G)$.  Many open questions remain regarding eternal domination with respect to these bounding parameters, and in particular whether or not particular constructions force $\edom(G)$ to be equal to one of them (the survey \cite{survey} provides a nice overview of what is and is not known about eternal domination and its variants).  

We make use of two graph binary operations on graphs in this paper.  The \textit{join} of two graphs $G$ and $H$, denoted $G \vee H$, is the graph obtained by adding all possible edges between $G$ and $H$.  The \textit{Cartesian product} of two graphs $G$ and $H$, denoted $G \Box H$, has $V(G \Box H) = V(G) \times V(H)$ and an edge between $(u,v)$ and $(x,y)$ if and only if either $u=x$ and $vy \in E(H)$ or $v=y$ and $ux \in E(G)$.
The product $G \Box K_2$ ($K_k$ denotes the complete graph on $k$ vertices) is called the \textit{prism} of $G$; one can informally think of the prism of $G$ as the graph obtained by taking two copies of $G$ and adding a matching between corresponding vertices.
In \cite{KM15}, the following conjecture is put forward:
\begin{conj}\cite{KM15} \label{theconj}
	If $G$ is a graph such that $\edom(G)=\theta(G)$, then $\edom(G \Box K_2)=\theta(G \Box K_2)$.
\end{conj}
The purpose of this paper is to present a construction of an infinite family of counterexamples to this conjecture.

\section{A Mycielskian construction}

Throughout the paper, we use $[n]$ to denote the set $\{1,2,3,\ldots,n\}$.  The \textit{Mycielskian} of a graph $G$, with $V(G) = \{v_1, v_2, \ldots, v_n\}$, denoted $M(G)$, is defined as follows:
    \begin{compactitem}
        \item[] $V(M(G)) = V(G) \cup \{v'_1, v'_2, \ldots, v'_n\} \cup \{x\}$
        \item[] $E(M(G)) = E(G) \cup E' \cup X$ where
            \begin{compactenum}
            	\item $E' = \{v_iv'_j \,|\, v_iv_j \in E(G)\}$ and
        		\item $X = \{v'_ix \,|\, i \in [n]\}$.
            \end{compactenum}
    \end{compactitem}

This construction was introduced in \cite{mycielski1955coloriage}, the purpose of which was to demonstrate the existence of triangle-free graphs with arbitrarily large chromatic number (recall that the \textit{chromatic number} of a graph $G$, denoted $\chi(G)$, is the minimum number of colours needed to colour $V(G)$ so that no adjacent vertices receive the same colour).  In particular, Mycielski proved the following:

\begin{lem}\cite{mycielski1955coloriage}\label{Mchi}
For any graph $G$, $\chi(M(G)) = \chi(G)+1$. 
\end{lem}

The following simple lemma follows directly from the construction of $M(G)$:

\begin{lem}\label{Momega}
For any connected graph $G$ on at least two vertices, $\omega(M(G)) = \omega(G)$. 
\end{lem}

Recall that a graph is called \textit{vertex-critical} if $\chi(G-v) < \chi(G)$ for every $v\in V(G)$.

\begin{lem}\label{Lem:Mvcritical}
    If $G$ is a vertex-critical graph, then $M(G)$ is vertex-critical.
\end{lem}

\begin{proof}
    Let $k = \chi(G)$ and let $c$ be a minimum proper $k$-colouring of $G$.  We proceed by cases on the vertex $v$ to be deleted from $M(G)$, showing in each case that $\chi(M(G)-v) = k$.  If $v = x$, then we assign $c(v'_i) = c(v_i)$.  Suppose $v = v'_i$.  Since $G$ is vertex critical, we may assume that $c(v_i)$ is a unique colour in $G$.  In this case $c$ can be extended to $M(G) - v'_i$ by letting $c(v'_j) = c(v_j)$ for $j \neq i$ and $c(x) = c(v_i)$.  Finally, suppose that $v = v_i$, and that $c(v_i)$ is unique in $G$.  In this case, $c$ can be extended by assigning $c(v'_j) = c(v_i)$ for all $j$, and $c(x)$ may be given any colour used in colouring $V(G)$ which is distinct from $c(v_i)$.
\end{proof}

We now define a particular family of graphs obtained by the Mycielskian operation.  Let $k \geq 2$ be an integer and let $\{\mathcal{M}_k^l\}_{l \geq k}$ be a family of graphs defined recursively by $\mathcal{M}_k^k = K_k$ and for each integer $l > k$, $\mathcal{M}_k^{l} = M\left(\mathcal{M}_k^{l-1})\right)$.
By Lemmas \ref{Mchi} and \ref{Momega}, $\chi(\mathcal{M}_k^l) = l$ and $\omega(\mathcal{M}_k^l) = k$.  By applying simple induction, we also have the following:

\begin{lem}\label{Lem:vcritical}
    For each $2 \leq k \leq l$, $\mathcal{M}_k^l$ is vertex-critical.
\end{lem}

\section{Main result}

The construction of our family of counterexamples to Conjecture \ref{theconj} requires the following three lemmas.

\begin{lem}\label{Lem:Alpha-Choose} \cite{klostermeyer2007eternal}
	For any graph $G$, $\edom(G) \leq {\alpha(G)+1 \choose 2}$.
\end{lem}

\begin{lem} \cite{KM05} \label{Lem:GraphJoin}
	Let $G$ be a graph such that $\alpha(G)=a$, $\gamma^\infty(G)=g$, $\theta(G)=t$.  If $p$ is an integer such that $g \leq p \leq t$, then $\alpha(G \vee \overline{K_p})=p$, $\gamma^\infty(G \vee \overline{K_p})=p$ and $\theta(G \vee \overline{K_p})=t$.
\end{lem}

\begin{lem} \label{Lem:G-Box-K2}
	Let $G$ be a graph and let $q_X$ denote the number of cliques of cardinality $1$ in a clique cover $X$ of $G$. If $q$ is the maximum value of $q_X$ taken over all clique covers of size $\theta(G)$, then $\theta(G \Box K_2)=2 \theta(G)-q$.
\end{lem}

\begin{proof}
	Without loss of generality, suppose the vertex set of $K_2$ is $\{1,2\}$. Let $\{C_1, C_2, \ldots, C_k\}$ be a clique cover of $G$ and let $C_j^i=\{(v,i): i \in [2], j \in [k], v \in C_j\}$. For any $i \in [2]$, if $G_i$ is the subgraph induced by the vertices of the set $V(G) \times \{i\}$ then $\{C_1^i, C_2^i, \ldots, C_k^i\}$ is a clique cover of $G_i$. For any $j \in [k]$, if $C_j=\{v\}$ then $\overline{C_j}=\{(v, 1), (v, 2)\}$ is a clique of $G \Box K_2$ and covers the vertices $(v, 1)$ and $(v, 2)$. Without loss of generality, suppose each of the cliques $C_1, C_2, \ldots, C_q$ covers a unique vertex and each of the cliques $C_{q+1}, C_{q+2}, \ldots, C_k$ have at least two vertices. The set $\{\overline{C_1}, \overline{C_2}, \ldots, \overline{C_q}\} \cup \{C_i^j: i \in \{1, 2\}, q+1 \leq j \leq k\}$ covers $G \Box K_2$, hence $\theta(G \Box K_2) \leq 2 \theta(G)-q$. Observe that a clique of $G \Box K_2$ is contained in either $G_1$ or $G_2$ or the clique is of the form $\overline{C}=\{(v, 1), (v, 2)\}$. The cliques contained in $G_1$ (respectively $G_2$) with the vertices $(v, 1)$ (respectively $(v, 2)$) in $\overline{C}$ cover $G_1$ (respectively $G_2$). Therefore, we then have the inverse inequality.
\end{proof}

We now use the construction of $\mathcal{M}_k^l$ to obtain our infinite family of counterexamples, which may have arbitrarily large independence number, eternal domination number, and clique cover number.

\begin{thm}
	For any integer $k \geq 2$, there exists a graph $G$ such that $\alpha(G)=\gamma^\infty(G)=\theta(G)={k+1 \choose 2}+1$ and $\alpha(G \Box K_2)=\gamma^\infty(G \Box K_2)<\theta(G \Box K_2)$.
\end{thm}

\begin{proof}
	Let $k$ be an integer greater than or equal to $2$, let $M = \mathcal{M}_k^{{k+1 \choose 2}+1}$, and let $H = \overline{M}$.  Since $\chi(M) = {k+1 \choose 2}+1$ and $\omega(M) = k$, we have that $\alpha(H)=k$ and $\theta(H)={k+1 \choose 2}+1$.  Since $M$ is vertex-critical (Lemma \ref{Lem:vcritical}), for any given vertex $w \in V(M)$ there exists a proper $\chi$-coloring of $V(M)$ such that $w$ receives a unique color.   As a consequence, for any given vertex $w \in H$, there exists a clique covering of $H$ by a minimum number of cliques where the clique which contains $w$ is of cardinality $1$.  By Lemma \ref{Lem:Alpha-Choose}, we have that $\edom(H) \leq {k+1 \choose 2}$.  Let $H^*=H \vee \overline{K}_{k+1 \choose 2}$; by Lemma \ref{Lem:GraphJoin}, $\alpha(H^*)=\gamma^\infty(H^*)={k+1 \choose 2}$ and $\theta(H^*)={k+1 \choose 2}+1$. Observe that for any given vertex $w \in V(H^*) \cap V(H)$, there exists a minimum clique covering of $H^*$ such that the clique which contains $w$ is of cardinality $1$, since each vertex of $\overline{K}_{k+1 \choose 2}$ can be added to a distinct clique of $H$ not containing $w$.  Fix $w$ to be any vertex in $H$, and note that $H^*$ has a maximum independent set $I$ for which $w \notin I$.  Finally, let $G$ be the graph obtained from $H^*$ by adding a single vertex $x$ adjacent only to $w$.   Then, $\alpha(G)=\alpha(H^*)+1={k+1 \choose 2}+1$.  By taking the minimum clique cover of $H^*$ for which $w$ is the only singleton, we may extend it to a clique cover of $G$ containing no singletons by replacing $\{w\}$ with $\{w,x\}$; by Lemma \ref{Lem:G-Box-K2}, this implies that $\theta(G \Box K_2)=2 {k+1 \choose 2}+2$.  On the other hand, denoting $V(K_2) = \{1,2\}$, ${k+1 \choose 2}$ guards can protect the subgraph induced by the vertices in the set $\{(v, 1): v \in V(H^*)\}$, ${k+1 \choose 2}$ guards can protect the subgraph induced by the vertices in the set $\{(v, 2): v \in V(H^*)\}$ and one guard can protect the vertices $(x, 1)$ and $(x, 2)$.  Thus $\edom(G \Box K_2) \leq 2 {k+1 \choose 2}+1 < \theta(G \Box K_2)$.
\end{proof}

\section{Acknowledgements}

Partial financial support for this work was received from the Fonds de recherche du Qu\'ebec-- Nature et technologies and from the Natural Sciences and Engineering Research Council of Canada.





\bibliographystyle{siam}
{\footnotesize \bibliography{sample}}







\end{document}